\documentclass[11pt]{amsart}
\usepackage{amsmath,amsfonts,amssymb,amsthm,epsf}
\usepackage{fullpage}

\newcommand{\be}{\begin{equation}}
\newcommand{\ee}{\end{equation}}
\newcommand{\bea}{\begin{eqnarray}}
\newcommand{\eea}{\end{eqnarray}}
\newcommand{\beas}{\begin{eqnarray*}}
\newcommand{\eeas}{\end{eqnarray*}}

\newtheorem{thm}{Theorem}
\newtheorem{corl}[thm]{Corollary}
\newtheorem{lma}[thm]{Lemma}
\newtheorem{prop}[thm]{Proposition}
\newtheorem{defn}[thm]{Definition}

\def\One{\mathbb{I}}

\def\bar{\overline}

\def\C{\mathbb{C}}

\newcommand{\ext}[1]{{(#1)}}

\def\I{{\rm I}}

\def\id{\mathrm{id}}

\def\prim{\textup{prim}}

\def\Sym{\mathrm{Sym}}

\def\tr{\mathrm{tr~}}

\def\fpt{\;\raisebox{-4mm}{\epsfysize=8mm\epsfbox{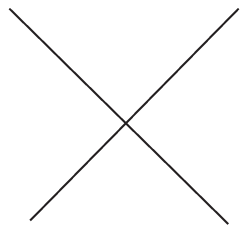}}\;}
\def\tpt{\;\raisebox{-4mm}{\epsfysize=8mm\epsfbox{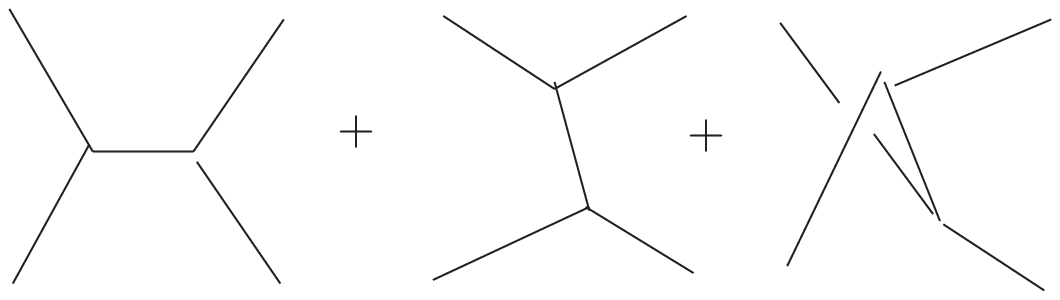}}\;}
\def\adscft{\;\raisebox{-36mm}{\epsfysize=100mm\epsfbox{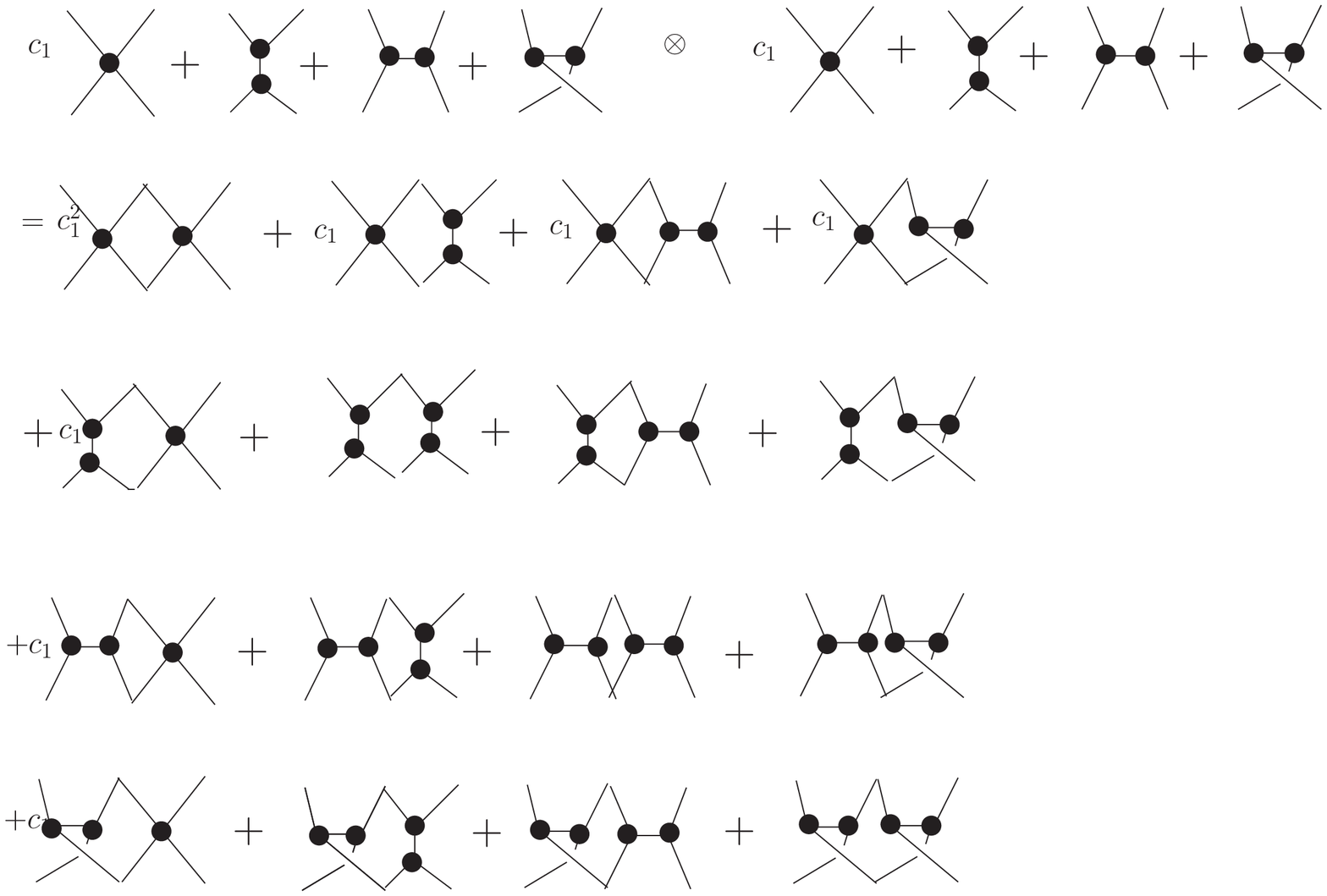}}\;}
\def\bcfw{\;\raisebox{-12mm}{\epsfysize=30mm\epsfbox{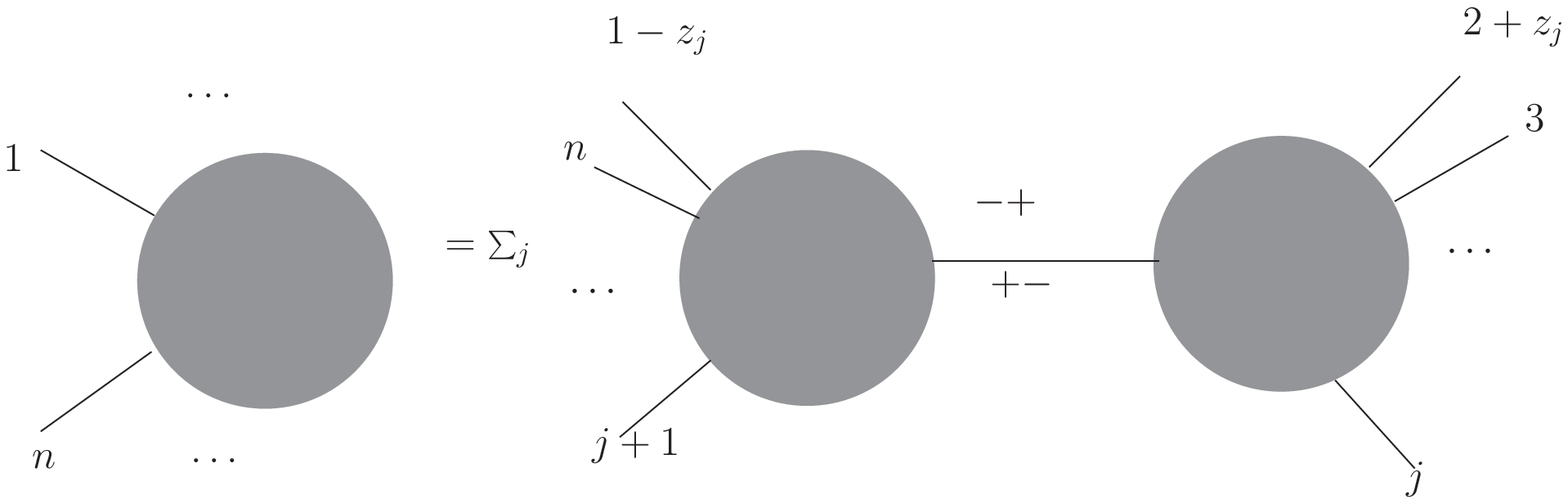}}\;}
\def\qedv{\;\raisebox{-2mm}{\epsfysize=6mm\epsfbox{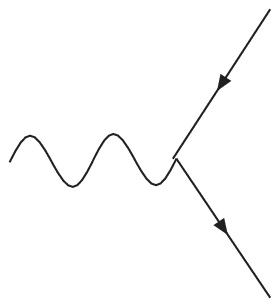}}\;}

\def\flptat{\;\raisebox{-2mm}{\epsfysize=5mm\epsfbox{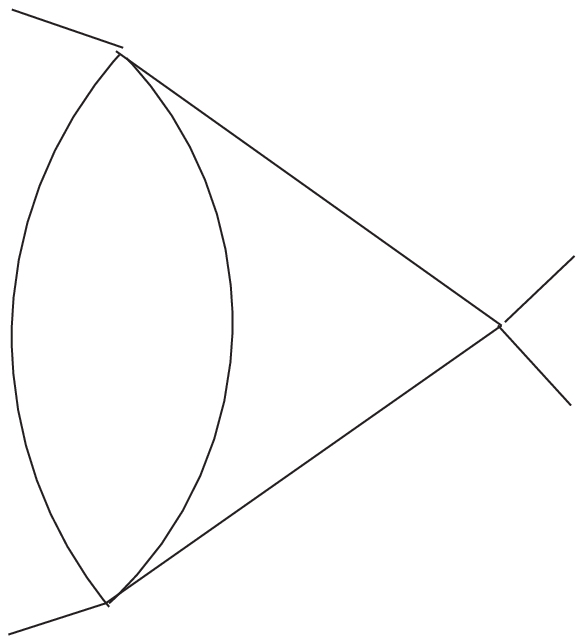}}\;}
\def\flptao{\;\raisebox{-2mm}{\epsfysize=5mm\epsfbox{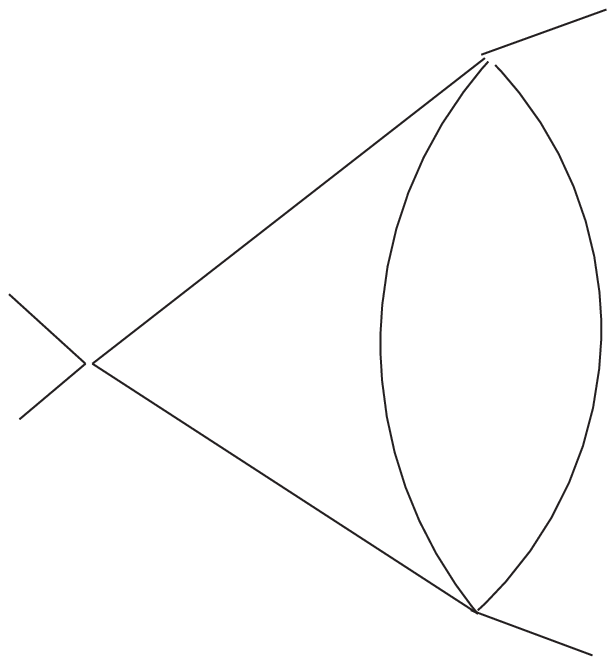}}\;}
\def\flptb{\;\raisebox{-1mm}{\epsfysize=3mm\epsfbox{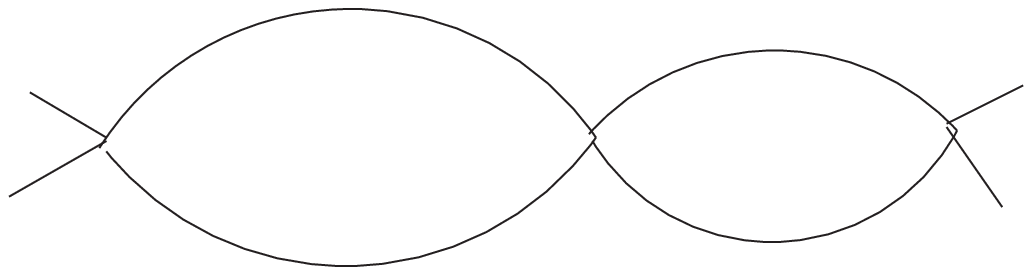}}\;}
\def\sol{\;\raisebox{-0.5mm}{\epsfysize=2mm\epsfbox{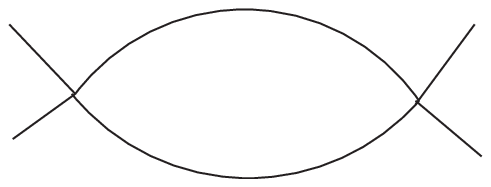}}\;}
\def\tol{\;\raisebox{-1mm}{\epsfysize=4mm\epsfbox{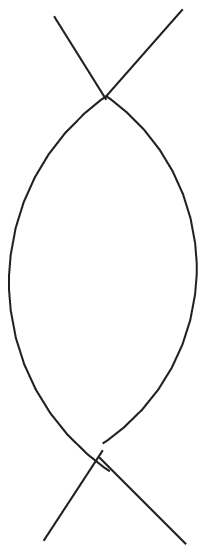}}\;}
\def\uol{\;\raisebox{-1mm}{\epsfysize=4mm\epsfbox{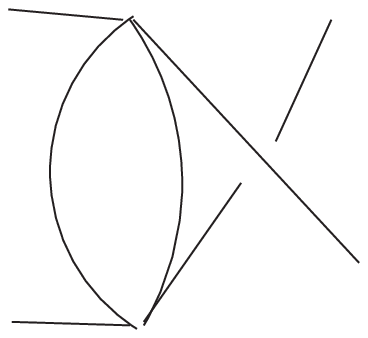}}\;}
\def\tadol{\;\raisebox{-2mm}{\epsfysize=4mm\epsfbox{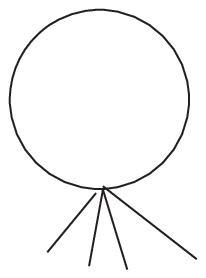}}\;}

\def\sixpt{\;\raisebox{-2mm}{\epsfysize=4mm\epsfbox{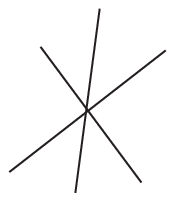}}\;}
\def\sixpto{\;\raisebox{-3mm}{\epsfysize=6mm\epsfbox{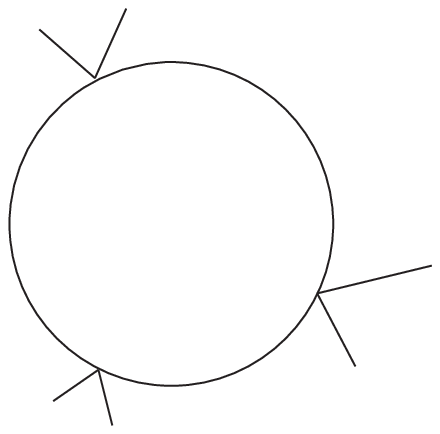}}\;}

\title{Recursive relations in the core Hopf algebra.}\thanks{This paper was partially supported by grant NSF grant DMS-0603781 at Boston University. DK supported by CNRS}
\author{Dirk Kreimer}
\address{IHES, 35 rte. de Chartres, 91440 Bures-sur-Yvette, France (http://\ www.ihes.fr)
and Boston U.\ (http://math.bu.edu), kreimer@ihes.fr}
\author{Walter D.\ van Suijlekom}
\address{Institute for Mathematics, Astrophysics and Particle Physics
Faculty of Science, Radboud Universiteit Nijmegen
Toernooiveld 1, 6525 ED Nijmegen, The Netherlands
waltervs@math.ru.nl}
\date{\today}

\begin{document}
\maketitle

\begin{abstract}
We study co-ideals in the core Hopf algebra underlying a quantum field theory.
\end{abstract}

\section{Introduction and Conventions}
In the following, we consider the core Hopf algebra of Feynman graphs.
It is a Hopf algebra which contains the renormalization Hopf algebra as a quotient Hopf algebra \cite{BK08}.
We are particularly interested in the structure of Green functions with respect to this Hopf algebra.

We write $G^r \equiv G^r(\{Q\},\{M\},\{g\};R)$ for a generic Green function, where
\begin{itemize}
 \item $r$ indicates the amplitude under consideration and we write $E\equiv |r|$ for its number of external legs. Amongst all possible amplitudes, there is a set
of amplitudes provided by the free propagators and vertices of the theory. We write $\mathcal{R}$ for this set. It is in one-to-one correspondence with field monomials in a Lagrangian approach to field theory. The set of all amplitudes is denoted by $\mathcal{A}=\mathcal{F}\cup\mathcal{R}$, which defines $\mathcal{F}$ as those amplitudes only present through quantum corrections.
\item $\{Q\}$ is the set of $E$ external momenta $q_j$ subject to the condition $\sum_{j=1}^E q_j=0$.
\item $\{M\}$ is the set of masses in the theory.
\item $\{g\}$ is the set of coupling constants specifying the theory. Below, we proceed for the case of a single coupling constant $g$,
the general case posing no principal new problems.
\item $R$ indicates the chosen renormalization scheme \cite{BK08}. \end{itemize}
We also note that a generic Green function $G^r$ has an expansion into scalar functions
\be G^r= \sum_{t(r)\in S(r)} t(r)G_{t(r)}^r(\{Q\},\{M\},\{g\};R).\ee
Here, $S(r)$ is a basis set of Lorentz covariants $t(r)$ in accordance with the quantum numbers specifying the amplitude $r$.
For each $t(r)\in S(r)$, there is a projector $P^{t(r)}$ onto this formfactor.

For example, in spinor quantum electrodynamics, the 1PI vertex function for the photon decay $p\to e^+e^- $ into a positron-electron pair $e^+(q)e^-(-q)$ measures the quantum corrections to that process described by a tree level vertex $$\gamma_\mu=t(\qedv)=\Phi(\qedv)$$
in terms of the Feynman rule $\Phi$ coming from the monomial $\bar{\psi}A\!\!\!/\psi$ in the QED Lagrangian.

At zero momentum for the photon $p$, computed in the momentum scheme $R_{\mathrm{mom}}$  that vertex function
can be decomposed in our notation as
\bea G^{\qedv} & = & \gamma_\mu G_{\gamma_\mu}^{\qedv}(\{q,-q\},m,e;R_{\mathrm{mom}})\nonumber\\ & &
+\frac{q\!\!\!/ q_\mu}{q^2} G_{\frac{q\!\!\!/ q_\mu}{q^2} }^{\qedv}(\{q,-q\},m,e;R_{\mathrm{mom}}),\eea
with projectors \be P^{\frac{q\!\!\!/ q_\mu}{q^2}}=  \frac{D}{D-1} \tr \circ \left(\frac{q\!\!\!/ q_\mu}{q^2}-\frac{1}{D}\gamma_\mu \right),\qquad P^{\gamma_\mu}=1-P^{\frac{q\!\!\!/ q_\mu}{q^2}}\ee in $D$ dimensions and where the trace is over the Dirac gamma matrices.

For $r\in\mathcal{R}$, we can write \be G^r=\Phi(r)G^r_{\Phi(r)}(\{Q\},\{M\},\{g\};R)+R^r(\{Q\},\{M\},\{g\};R),\ee
where $R^r(\{Q\},\{M\},\{g\};R)$  sums up all formfactors $t(r)$ but $\Phi(r)$ and only contributes through quantum corrections, and $\Phi$ are the unrenormalized Feynman rules.

Each $G^r$ can be obtained by the evaluation of a series of 1PI graphs
\bea
X^r(g) & = & \One - \sum_{E(\Gamma)\sim r}g^{|\Gamma|}\frac{\Gamma}{\mathrm{Sym}(\Gamma)},\forall r\in \mathcal{R},|r|=2,\label{green1}\\
X^r(g) & = & \One + \sum_{E(\Gamma)\sim r}g^{|\Gamma|}\frac{\Gamma}{\mathrm{Sym}(\Gamma)},\forall r\in \mathcal{R},|r|>2,\label{green2}\\
X^r(g) & = & \sum_{E(\Gamma)\sim r}g^{|\Gamma|}\frac{\Gamma}{\mathrm{Sym}(\Gamma)},\forall r\notin \mathcal{R},\label{green3}
\eea
where we take the minus sign for $|r|=2$ and the plus sign for $|r|>2$. Furthermore, the notation $E(\Gamma)\sim r$ indicates a sum over graphs with external leg structure in accordance with $r$.

We write $\Phi,\Phi_R$ for the unrenormalized and renormalized Feynman rules regarded as a map: $H\to \mathbb{C}$
from the Hopf algebra to  $\mathbb{C}$. In a slight abuse of notation, we use the same symbol to denote the map which assigns to an element
of $\mathcal{R}$ the corresponding Lorentz covariant, as in $\Phi(\qedv)=\gamma_\mu$.

We have
\be G^r_{t(r)}=\Phi_R^{t(r)}(X^r(g))(\{Q\},\{M\},\{g\};R),\ee
where each non-empty graph is evaluated by the  renormalized Feynman rules \be\Phi_R^{t(r)}(\Gamma):=(1-R)m(S_R^\Phi \otimes P^{t(r)}\Phi P)\Delta(\Gamma)\label{renFR}\ee
and $\Phi_R^{t(r)}(\One)=1$, and $P$ the projection into the augmentation ideal of $H$, and $R$ the renormalization map.

It is in the evaluation (\ref{renFR}) that the coproduct of the renormalization Hopf algebra appears.

The above sum over all graphs simplifies when one takes the Hochschild cohomology of the (renormalization) Hopf algebra into account:
\be X^r(g)=\delta_{r,\mathcal{R}}\One\pm \sum_{E(\gamma)\sim r;\Delta(\gamma)=\gamma\otimes \One+\One\otimes\gamma}\frac{1}{\mathrm{Sym}(\Gamma)}g^{|\gamma|}  B_+^\gamma(X^r(g)Q(g)),\ee
($-$ sign for $|r|=2$, + sign for $|r|>2$, $\delta_{r,\mathcal{R}}=1$ for $r\in \mathcal{R}, 0$ else)
with $Q(g)$ being the formal series of graphs assigned to an invariant charge of the coupling $g$:
\be Q^r(g)=\left[\frac{\overline{X}^{r,|r|>2}}{\prod_{e\in E(r)}\sqrt{X^e}}\right]^{\frac{1}{|r|-2}},\ee
where $\overline{X}^{r}=X^r$ for $r\in\mathcal{R}$ and $\overline{X}^{r}=X^r+\One$ else. Also, $B_+^\gamma$ are {\it grafting operators} which are Hochschild cocycles (cf. Section \ref{sect:ds} below).

Note that the existence of a unique such invariant charge  depends on the existence of suitable coideals in the renormalization Hopf algebra as discussed below.

There is a tower of quotient Hopf algebras
\be H_4\subset H_6\cdots\subset H_{2n}\cdots H_{\mathrm{core}}=H,\ee
obtained by restricting the coproduct to sums over graphs which are superficially divergent in $$D=4,6,\ldots,2n,\ldots,\infty$$ dimensions. They are defined via a coproduct which restricts to superficially divergent graphs $\omega_D(\Gamma)\leq 0$ in an even number of dimension $D$ greater than the critical dimension $D=4$.
\be
\Delta (\Gamma) = \Gamma \otimes 1 + 1 \otimes \Gamma + \sum_{\emptyset \subsetneq \gamma \subsetneq \Gamma; \omega_D} \gamma \otimes \Gamma/\gamma,
\ee
where $\omega_D$ restricts to disjoint unions $\gamma=\cup_i\gamma_i$ such that $\omega_D(\gamma_i)\leq 0$ for all $\gamma_i$.

\subsection{Remarks}
The above algebraic structures given by this tower of Hopf algebras underly many familiar aspects of field theory.
For example, in effective field theories, one consider couplings for any interaction in accordance with the symmetries of the theory, often suppressed by a scale which is large compared to the scales which are experimentally observable. The set $\mathcal{R}$ can then exhaust the full set $\mathcal{A}$.
Still, the Hopf algebra renormalizing the corresponding Green functions is a quotient Hopf algebra of the core Hopf algebra, as some amplitudes
might only demand counterterms from a suitably high loop number onwards.

Also, in operator product expansions we effectively enlarge the set $\mathcal{R}$ to contain any local amplitude which appears in the high-momentum Taylor expansion of a given amplitude in terms of local operator insertions, and the corresponding renormalizations in this expansion form again a quotient Hopf algebra of the core Hopf algebra.

Finally, for theories which obey a gravity power-counting, the core Hopf algebra becomes the renormalization Hopf algebra, and the corresponding co-ideal structure \cite{Kre07} is suggesting hidden renormalizability, in accordance with the recursive relations between on-shell gravity scattering amplitudes \cite{B-V07}.
\section{The core Hopf algebra}
We start by recalling the definition of the core Hopf algebra \cite{BK08} (cf. also \cite{Kre09}) built on Feynman graphs with only a scalar edge.
Besides that, we allow vertices of any valence to appear.
We omit the general case involving vertices or edges of different kinds for clarity of notation.

Recall that a {\it one-particle irreducible (1PI) graphs} is a graph which is not a tree and does not become disconnected when cutting a single internal edge.

We will use the following notation for a Feynman graph:

\medskip

\begin{tabular}{p{1cm}l}
$E, I$ & the number of external and internal lines, respectively;\\
$V_n$ &the number of vertices of valence $n$, which sum up to the total number of vertices $V$;\\
$L$ &the number of loops.
\end{tabular}

\medskip

\begin{lma}
\label{lma:degrees}
There are the following relations between these numbers:
\begin{gather*}
2I + E = \sum_n n V_n ;
\qquad
\sum_n (n-2)V_n - (E-2)= 2L.
\end{gather*}
\end{lma}
\begin{proof}
The first equation follows after realizing that the left-hand-side counts the number of halflines in a graphs, which are connected to $V_n$ vertices of valence $n$, for each $n$, appearing at the right-hand-side. Moreover, subtracting twice Euler's formula $I-V+1 = L$ from it gives the second displayed equation.
\end{proof}
The above relations turn out to be quite useful later on. Let us now turn to the definition of the core Hopf algebra.

\begin{defn}
The {\rm core Hopf algebra} $H$ is the free commutative algebra (over $\C$) generated by all 1PI Feynman graphs with counit $\epsilon(\Gamma)=0$ unless $\Gamma=\emptyset$, in which case $\epsilon(\emptyset)=1$, coproduct,
\begin{align*}
\Delta (\Gamma) = \Gamma \otimes 1 + 1 \otimes \Gamma + \sum_{\emptyset \subsetneq \gamma \subsetneq \Gamma} \gamma \otimes \Gamma/\gamma,
\end{align*}
where the sum is over all disjoint unions of 1PI (proper) subgraphs in $\Gamma$. Finally, the antipode is given recursively by,
\begin{equation}
\label{antipode}
S(\Gamma) = - \Gamma - \sum_{\emptyset \subsetneq \gamma \subsetneq \Gamma} S(\gamma) \Gamma/\gamma.
\end{equation}
\end{defn}
Even though the graphs $\Gamma$ can have vertices of arbitrary valence (as opposed to the usual Feynman graphs in renormalizable perturbative quantum field theories), the Hopf algebra structure is still well-defined. Indeed, in view of the above Lemma, at a given loop order $L$ and number of external lines $E$, the maximal vertex valence that appears in the graph is finite.

The Hopf algebra is graded by loop number, since the number of loops in a subgraph $\gamma \subset \Gamma$ and in the graph $\Gamma/\gamma$ add up to $L(\Gamma)$. Another multi-grading is given by the number of vertices. In order for this to be compatible with the coproduct -- creating an extra vertex in the quotient $\Gamma/\gamma$ -- we say a graph $\Gamma$ is of multi-vertex-degree $\bar k = (k_3, k_4, \ldots)$ if
$$
V_n (\Gamma) = k_n + \delta_{n,E(\Gamma)}.
$$
One can check that this grading is compatible with the coproduct. From Lemma \ref{lma:degrees} if follows easily that the two degrees are related via $\sum_m (m-2) k_m = 2L$.

From a physical point of view, it is not so interesting to study individual graphs; rather, one considers whole sums of graphs with the same number of external lines. Namely, we study the {\it 1PI Green's functions} of (\ref{green1},\ref{green2},\ref{green3}) as elements in $H$.

Regarding the above gradings, we denote the above sum when restricted to graphs with $l$ loops by $X_l^{r,|r|=n}$. Also, the restriction of $X^{r,|r|=n}$ to graphs with $k_m + \delta_{m,n}$ vertices of valence $m$ ($m=3,4,\ldots$) will be written as $X^{r, n}_{\bar k}$ with $\bar k = (k_3, k_4, \ldots )$ as before.
\section{Hopf ideals in $H$}
In this section we address the question of how the coproduct acts on the above Green's functions $X^{r,|r|=n}$. From \cite{Sui07b} we take the following
\begin{prop}
The coproduct reads on the 1PI Green's functions ($n \geq 2$):
$$
\Delta(X^{r,|r|=n}) =
\sum_{E(\Gamma)\sim r} \prod_{m} \left[ X^{r,|r|=m} \right]^{V_m(\Gamma)} \left[ X^{r,|r|=2} \right]^{-I(\Gamma)} \otimes \frac{g^{|\Gamma|}\Gamma}{\Sym(\Gamma)}.
$$
\end{prop}
\begin{corl}
\label{corl:cop-green}
The coproduct takes the following form on the 1PI Green's functions ($|r| \geq 2$):
$$
\Delta(X^{r,|r|=n}) = \sum_{\bar k} X^{r,|r|=n} \prod_{m=3}^\infty \left[ \frac{X^{r,|r|=m} }{\left(X^{r,|r|=2} \right)^{m/2}} \right]^{k_m}   \otimes X^{r,|r|=n}_{\bar k}.
$$
\end{corl}
\begin{proof}
This follows easily by applying the first Equation in Lemma \ref{lma:degrees}, in combination with the definition of the multigrading $k_m$.
\end{proof}
Next, we define the following `couplings' $Q^\ext{m}$ in $H$
$$
Q^\ext{m} = \left[ \frac{\overline{X}^{r,|r|=m} }{\left(X^{r,|r|=2} \right)^{m/2}} \right]^{1/({m-2)}}; \qquad (m>2),
$$
which, when restricted to loop order $l$, are denoted by $Q^\ext{m}_l$.
\begin{prop}
The ideal $I = \langle Q^\ext{m}_l - Q^\ext{n}_l \rangle$ where $m,n \geq 3$ and $l \geq 0$ is a Hopf ideal, {i.e.}
$$
\Delta(I) \subset I \otimes H + H \otimes I, \qquad \epsilon(I) = 0, \qquad S(I) \subset I.
$$
\end{prop}
\begin{proof}
From the relation $\sum_m (m-2)k_m=2l$ between the multigrading $k_m$ by number of vertices and the loop order $l$, it follows that we have
$$
\Delta(X^{r,|r|=n})= \sum_l X^{r,|r|=n} \left( Q^\ext{3} \right)^{2l} \otimes X_l^{r,|r|=n} + I \otimes H.
$$
Indeed, modulo $I$, one can replace each $(Q^\ext{m})^{(m-2)k_m}$ by $(Q^\ext{3})^{(m-2)k_m}$. This leads precisely to $\sum_m (m-2)k_m=2l$ factors of $Q^\ext{3}$. Extending this to formal powers (such as $\frac{-m}{2(m-2)}$ of Green's functions appearing in the couplings $Q^\ext{m}$)) this leads to
$$
\Delta(Q^\ext{n})= \sum_l Q^\ext{n} \left( Q^\ext{3} \right)^{2l} \otimes Q_l^\ext{n} + I \otimes H
$$
from which the claim follows.
\end{proof}
This implies that the quotient $H/I$ is a Hopf algebra in which the relations $Q^\ext{m} = Q^\ext{n}$ hold; we will also set $Q= Q^\ext{m}$. A recursive way of writing these relations is
\be
\frac{ X^{r,|r|=n}} { X^{r,|r|=n-1}}  = \frac{X^{r,|r|=n-1}} {X^{r,|r|=n-2}} .\label{STcore}
\ee
In a non-abelian gauge theories the above identities actually hold between the corresponding physical amplitudes so that Feynman rules provide an element of $\mathrm{Spec}(H/I)$ for $n=4$
and are kwown as the Slavnov--Taylor identities for the couplings. Note that then countertems in a chosen renormalization scheme furnish an element in $\mathrm{Spec}(H/I)$. For renormalized amplitudes, the choice of a renormalization condition for any of the vertices in $\mathrm{R}$ then determines
Feynman rules in $\mathrm{Spec}(H/I)$ for any such choice. See \cite{CelmGon} for an excellent discusion of such choices.

We conclude this section with an expression for the coproduct on Green's functions in the quotient.
\begin{prop}
\label{prop:cop-quotient}
In the quotient Hopf algebra $H/I$ we have
$$
\Delta(X^{r,|r|=n}) = \sum_l X_{l,n} \otimes X_l^{r,|r|=n}
$$
where we have denoted $X_{l,n} = X^{r,|r|=n} Q^{2l}$.
\end{prop}
\begin{corl}
\label{corl:cop-Xln}
$$
\Delta(X_{l,n}) = \sum_j X_{l+j,n} \otimes (X_{l,n})_j.
$$
\end{corl}
\section{Hopf subalgebras and Dyson--Schwinger equations}
\label{sect:ds}
Another way to describe the Green's function is in terms of so-called grafting operators, defined in terms of 1PI primitive graphs.
We start by considering maps $B_+^\gamma: H\to \mathrm{Aug}$, with $\mathrm{Aug}$ the augmentation ideal, which will soon lead us to non-trivial one co-cycles in the Hochschild cohomology of $H$. They are defined as follows.
\be B_+^\gamma(h)=\sum_{\Gamma\in <\Gamma>}\frac{{\textbf{bij}(\gamma,h,\Gamma)}}{|h|_\vee}\frac{1}{\textrm{maxf}(\Gamma)}\frac{1}{(\gamma|h)}\Gamma,\label{def}\ee
where maxf$(\Gamma)$ is the number of maximal forests of $\Gamma$, $|h|_\vee$ is the number of distinct graphs obtainable by permuting edges of $h$, $\textbf{bij}(\gamma,h,\Gamma)$ is the number of bijections of external edges of $h$ with an insertion place in $\gamma$ such that the result is $\Gamma$,
and finally $(\gamma|h)$ is the number of insertion places for $h$ in $\gamma$ \cite{Kre05}. $\sum_{\Gamma\in <\Gamma>}$ indicates a sum over the linear span $<\Gamma>$ of generators of $H$.

The sum of the $B^\gamma_+$ over all primitive 1PI Feynman graphs at a given loop order and with given residue will be denoted by $B^{l;n}_+$, as in {\it loc. cit.}. More precisely,
$$
B^{l;n}_+ = \sum_{\begin{smallmatrix} \gamma ~\prim \\ l(\gamma)=l \\ E(\gamma)=n \end{smallmatrix}} \frac{1}{\Sym(\gamma)} B^\gamma_+.
$$
With this and the formulas of the previous section on QCD, we can prove the analog of the {\it gauge theory theorem} as formulated in \cite[Theorem 5]{Kre05}:
\begin{thm}
\label{thm:gauge}
Let $\tilde H=H/I$ be the core Hopf algebra with relations $Q_\ext{n} = Q_\ext{m}$ as before.
\begin{enumerate}
\item $X^{r,|r|=n} = \sum_{l=0}^\infty B_+^{l;n} (X_{l,n})$.
\item $\Delta( B_+^{l;n} (X_{l,n})) = B_+^{l;n} (X_{l,n}) \otimes \I + (\id \otimes B_+^{l;n}) \Delta (X_{l,n})$.
\item $\Delta( X_l^{r,|r|=n}) = \sum_{j=0}^l \textup{Pol}^n_j(X) \otimes X_{l-j}^\ext{r,|r|=n}$.
\end{enumerate}
where $\textup{Pol}^r_j(X)$ is a polynomial in the $X^{r,|r|=n}_m$ of degree $j$, determined as the order $j$ term in the loop expansion of $X^{r,|r|=n} Q^{2l-2j}$.
\end{thm}
\begin{proof}
The first claim follows as in \cite{Kre05}. $B_+^\gamma$ acts on arguments which have multiplicity $(\gamma|h)\times|h|_\vee$.
We hence have to divide by this multiplicity, and by the number $\mathrm{maxf}(\Gamma)$ of ways to generate $\Gamma$.
This by construction generates any graph with weight $1/\mathrm{Sym}(\Gamma)$ \cite{Kre05}. Indeed, assume for a moment that we label the external edges of $h$ and internal edges of $\gamma$ and that we keep those labels in the bijections which define $\Gamma$. Then each labeled graph is generated once, the bijections define an operadic composition and the assertion follows as in the operadic proof of Lemma 4 of \cite{BergKr06}.

For the second claim, we first enhance the result of Corollary \ref{corl:cop-green} to partial sums in $X^{r,|r|=n}$ over graphs that have `primitive residue' isomorphic to a fixed primitive graph $\gamma$. In other words, if $X^{{r,|r|=n}, \gamma}$ is the part of $X^{r,|r|=n}$ that sums only over graphs that are obtained by inserting graphs into the primitive graph $\gamma$, then
$$
\Delta( X^{{r,|r|=n},\gamma} ) = X^{{r,|r|=n},\gamma} \otimes 1 + \sum_{l=1}^\infty X_{l,n} \otimes ( X^{{r,|r|=n},\gamma}_l ).
$$
Here we have imposed (\ref{STcore}) to write this in terms of the single coupling $Q$.
Since $X^{{r,|r|=n},\gamma} = B^\gamma_+ ( X_{l,n} )$, a combination of this formula with Corollary \ref{corl:cop-Xln} yields
$$
\Delta(B^\gamma_+ ( X_{l,n} )) = B^\gamma_+ ( X_{l,n}) \otimes \I + (\id \otimes B^\gamma_+) \Delta( X_{l,n}).
$$
Then, summing over all primitive graphs with $n$ external lines at loop order $l$ gives the desired result.
\end{proof}
Note that our formulation above is somewhat redundant: in the core Hopf algebra all primitives have a single loop, $|\gamma|=1$.
But in the given form, the results remain applicable to any of the before-mentioned quotient Hopf algebras, where primitives appear beyond the first order still.

In fact, this proves the slightly stronger result that every $B_+^\gamma$ defines a Hochschild 1-cocycle:
\begin{prop}
For $\gamma$ a primitive graph at loop order $k$ and residue $r$, we have
$$
\Delta( B_+^\gamma (X_{l,n})) = B_+^\gamma (X_{l,n}) \otimes \I + (\id \otimes B_+^\gamma) \Delta (X_{l,n}).
$$
\end{prop}
\section{An Example}
Let us study $\phi^4$ theory in four dimensions and work out the 1PI vertex function to two loops.
We have
\beas X^{\times,|\times|=4} & = & \One+g\frac{1}{2}\left[\sol+\tol+\uol\right]+g^2\left[\frac{1}{2}\left(\flptao+\flptat\right)+\frac{1}{4}\flptb\right]\\ & & +\mathrm{other \;orientations}+\mathcal{O}(g^3).\eeas
Let us reproduce this expansion from either the Hochschild cohomology of the renormalization Hopf algebra
or from the Hochschild cohomology of the core Hopf algebra.
In the renormalization Hopf algebra we have
\be X^{\times,|\times|=4}=\One+B_+^{1,\times}([X^{\times,|\times|=4}]Q)\ee
(there is no primitive element at two-loops)
with
\be B_+^{1,\times}=\frac{1}{2}\left(B_+^{\sol}+B_+^{\tol}+B_+^{\uol}\right)\ee
and
\be Q=\frac{X^{\times,|\times|=4}}{\left[X^{-,|-|=2}\right]^2}.\ee
In the core Hopf algebra
\be X^{\times,|\times|=4}=\One+B_+^{1,\times}(X^{\times,|\times|=4}Q)\ee
but now
\be Q_1:=\frac{X^{\times,|\times|=4}}{\left[X^{-,|-|=2}\right]^2}
=\left[ \frac{\overline{X}^{\sixpt,\left|\!\sixpt\!\right|=6}}{\left[X^{-,|-|=2}\right]^3}\right]^{\frac{1}{2}}=:Q_2,\ee
where the core co-ideal implies $Q_1=Q=Q_2$.
Furthermore
\be B_+^{1,\times}=\frac{1}{2}\left(B_+^{\sol}+B_+^{\tol}+B_+^{\uol}+B_+^{\tadol}\right).\ee

The coproduct is different in the renormalization Hopf algebra and in the core Hopf algebra.
In the renormalization Hopf algebra we find
\be \Delta^\prime\left[\frac{1}{2}\left(\flptao+\flptat\right)+\frac{1}{4}\flptb\right]=\frac{1}{2}\left[\sol+\tol+\uol\right]\otimes\left[\sol+\tol+\uol\right],\ee
taking the other orientations into account.

In the core Hopf algebra we find
\be \Delta_c^\prime \left[\frac{1}{2}\left(\flptao+\flptat\right)+\frac{1}{4}\flptb\right] =\frac{1}{2}\left[\sol+\tol+\uol\right]\otimes\left[\sol+\tol+\uol\right]+2\sixpto\otimes\tadol.\ee

Let us now work out the Hochschild one-cocycles. Expanding the arguments of the $B_+$ cocycles to one loop and keeping the vertex functions to one loop, we have
\be \frac{1}{2}B_+^{\sol}(2X_1^{\times,|\times=4|})\ee
in the renormalization Hopf algebra and
\be \frac{1}{2}(B_+^{\sol}(2X_1^{\times,|\times|=4})+B_+^{\tadol}(X_1^{\sixpt,\left|\!\sixpt\!\right|=6}))\ee
in the core Hopf algebra.
Let us see in particular how the terms
\be \frac{1}{2}\left(\flptao+\flptat\right)+\frac{1}{4}\flptb\label{flpt}\ee are obtained in either case.
We have in the renormalization Hopf algebra
\beas
{\textbf{bij}}(\sol,\sol,\flptb)& = & 1,\\
{\textbf{bij}}(\sol,\sol,\flptao)& = & 2,\\
{\textbf{bij}}(\sol,\sol,\flptat)& = & 2,\\
\left|\sol\right|_\vee & = & 3,\\
\left(\sol|\sol\right) & = & 2,\\
\mathrm{maxf}\left(\flptb\right) & = & 2,\\
\mathrm{maxf}\left(\flptao\right) & = & 1,\\
\mathrm{maxf}\left(\flptat\right) & = & 1.
\eeas

The main difference is in the number of maximal forests $\mathrm{maxf}$.
In the renormalization Hopf algebra, the first two terms in (\ref{flpt}) have just one maximal forest, while the third has two, as indicated.
In the core Hopf algebra the first two terms have three maximal forests each, while the third term has two as before.

We hence find, counting maximal forests, bijections, insertion places and orientations, as above,
\be \frac{1}{2}B_+^{\sol}(2\times \frac{1}{2}\left[\sol+\tol+\uol\right])=\frac{1}{2}\left(\flptao+\flptat\right)+\frac{1}{4}\flptb,\ee
while in the core Hopf algebra, the situation is a bit more interesting:
\be \frac{1}{2}B_+^{\sol}(2\times \frac{1}{2}\left[\sol+\tol+\uol\right])=\frac{1}{6}\left(\flptao+\flptat\right)+\frac{1}{4}\flptb,\ee
\be \frac{1}{2}B_+^{\tadol}(\sixpto)=\frac{1}{3}\left(\flptao+\flptat\right),\ee
which add up to the desired result, also confirming the cocycle property of the $B_+$ maps.
\section{The core ideals, unitarity and AdS/CFT}
We conclude this paper with a short study of the relation between the core co-ideal $I$ introduced above and recursive relations between tree-level amplitudes, as suggested in \cite{Kre09}.

In the above, we identified a core co-ideal conveniently summarized by the relations
\bea \frac{X^{r,|r|=n+1}}{X^{r,|r|=n}} =  \frac{X^{r,|r|=n}}{X^{r,|r|=n-1}} & & \Leftrightarrow\nonumber\\
\Leftrightarrow X^{r,|r|=n} =  X^{r,|r|=j}\frac{1}{X^{r,|r|=2}}X^{r,|r|=k} & &
,\forall n>2,j>2,k>2,j+k=n+2.\label{adscftideal}\eea

Note that this severely restricts possible relations between tree-level diagrams.
Consider for example at zero loops the tree graphs
\be T_0(c_1):=P^{\Phi(r),|r|=4}\Phi\left(c_1\fpt+\tpt\right).\ee
The projector $P^{\Phi(r),|r|=4}$ maps the evaluation of the tree level diagrams, for given fixed external momenta,  to complex numbers, and vanishes on none of the four terms. There must hence exist a number $c_1$ such that $T_0(c_1)=0$.
Let $T(c_1)=\sum_{j\geq 0}T_j(c_1)$ be the expansion obtained by a loop expansion of any internal vertex or propagator in
$T_0(c_1)$.

We can now determine $c_1$ from squaring the amplitude $T_0(c_1)$.
This delivers (for the $s$-channel)
\be \adscft.\ee
This is in accordance with the co-ideal if and only if $c_1=-1$. Indeed, restricting to the one-loop case
we have in the co-ideal \be 2 X_1^{r,|r|=3}+X_1^{r,|r|=2}=X_1^{r,|r|=4},\ee
which is consistent with the expansion of $T(c_1)$ to one-loop only at $c_1=-1$.
This is clearly seen in the figure. Even the 1PI graphs in the $s$-channel expansion of \be \fpt,\ee come with different powers of $c_1$.
On the other hand, at $c_1=-1$, summing over $s,t,u$ channels and taking symmetry factors into account, we find
\be X_1^{r,|r|=4}=\left[X^{r,|r|=3}\frac{1}{X^{r,|r|=2}}X^{r,|r|=3}\right].\ee

While in the renormalization Hopf algebra one uses only the co-ideal for $|r|=4$, in the core Hopf algebra we have a generic
co-ideal structure such that the celebrated BCFW recursion is required for consistency of the Feynman rules with that co-ideal structure beyond tree-level.

\begin{prop}
The relations
\be \bcfw, \label{bcfw}\ee
are in accordance with the  co-ideal above for suitably chosen elements in the group $\mathrm{Spec}(H/I)$, defined by evaluating external particles on the mass-shell in accordance with the BCFW rules \cite{BCF04,BCFW05,B-V07}.
\end{prop}

\begin{proof}
We use (\ref{adscftideal}). The loop expansion of the vertices and the internal propagator in each term of (\ref{bcfw}) on the lhs is an expansion of $X^{r,|r|=n}$, the expansion of the rhs is an expansion of $X^{r,|r|=n-j+1} \frac{1}{ X^{r,|r|=2} } X^{r,|r|=j}$.
The choices of helicities at all external propagators, and the internal propagator, and the mass-shell conditions for all external legs,
specify then an element in the group $\mathrm{Spec}(H/I)$ by which we evaluate those graphs.
\end{proof}
Note that this does not pretend we can derive the recursion (\ref{bcfw}) from our co-ideal. It merely says that those recusions are in accordance with the most natural co-ideal in the core Hopf algebra.
\section{Outlook}
In this paper we have given the core Hopf algebra as a mathematically robust framework to investigate many properties of  Feynman rules
in $\mathrm{Spec}(H/I)$ which emerge in the recent literature. We hope that the first steps reported here open the way to a much better understanding of recursive relations between multi-leg and multi-loop amplitudes.
\section*{Acknowledgments} 
WvS thanks Thomas Quella for discussions on MHV recursions and Hopf algebras. DK thanks Karen Yeats for numerous discussions regarding the core Hopf algebra. WvS thanks IHES for hospitality.


\begin{thebibliography}{10}

\bibitem{BergKr06}
  C.~Bergbauer and D.~Kreimer.
  \newblock {Hopf algebras in renormalization theory: Locality and Dyson-Schwinger
  equations from Hochschild cohomology}.
  \newblock{\em IRMA Lect.\ Math.\ Theor.\ Phys.\  10} (2006) 133-164.

\bibitem{BK08}
S.~Bloch and D.~Kreimer.
\newblock {Mixed Hodge Structures and Renormalization in Physics}.
\newblock {\em Commun. Num. Theor. Phys. 2.4} (2008)  637--718.

\bibitem{B-V07}
C.~Boucher-Veronneau.
\newblock {On-Shell recursion relations in general relativity}.
\newblock {\em Master Thesis, U.\ of Waterloo, Ontario, Canada 2007}.

\bibitem{BCF04}
R.~Britto, F.~Cachazo and  B.~Feng.
\newblock {New recursion relations for tree amplitudes of gluons}.
\newblock {\em Nucl.Phys.B} 715:499-522,2005.

\bibitem{BCFW05}
R.~Britto, F.~Cachazo, B.~Feng and E.~Witten.
\newblock {Direct proof of tree-level recursion relation in Yang-Mills theory}.
\newblock {\em Phys.Rev.Lett.} 94:181602,2005.

\bibitem{CelmGon}
  W.~Celmaster and R.~J.~Gonsalves.
\newblock {The Renormalization Prescription Dependence Of The QCD Coupling Constant}.
\newblock {\em Phys.\ Rev.\  D} { 20} (1979) 1420.

\bibitem{CK00}
A.~Connes and D.~Kreimer.
\newblock Renormalization in quantum field theory and the {R}iemann- {H}ilbert
  problem. {II}: The beta-function, diffeomorphisms and the renormalization
  group.
\newblock {\em Commun. Math. Phys.} 216 (2001)  215--241.

\bibitem{CK98}
A.~Connes and D.~Kreimer.
\newblock {Hopf algebras, renormalization and noncommutative geometry}.
\newblock {\em Commun. Math. Phys.} 199 (1998)  203--242.

\bibitem{Kre05}
D.~Kreimer.
\newblock Anatomy of a gauge theory.
\newblock {\em Ann. Phys.} 321 (2006)  2757--2781.

\bibitem{Kre09}
D.~Kreimer.
\newblock The core {H}opf algebra.
\newblock To appear in Proceedings for Alain Connes.

\bibitem{Kre07}
D.~Kreimer.
\newblock {A remark on quantum gravity}.
\newblock {\em Ann. Phys.} 323 (2008)  49--60.

\bibitem{Kre08}
D.~Kreimer.
\newblock {Not so non-renormalizable gravity}.
\newblock{In {\em Quantum Field Theory: Competitive Models,}} B.Fauser, J.Tolksdorf, E.Zeidler, eds., Birkhaeuser (2009),
155-162.

\bibitem{KY06}
D.~Kreimer and K.~Yeats.
\newblock An \'etude in non-linear {D}yson-{S}chwinger equations.
\newblock {\em Nuclear Phys. B Proc. Suppl.} 160 (2006)  116--121.

\bibitem{Sui07}
W.~D. van Suijlekom.
\newblock Renormalization of gauge fields: {A} {H}opf algebra approach.
\newblock {\em Commun. Math. Phys.} 276 (2007)  773--798.

\bibitem{Sui07b}
W.~D. van Suijlekom.
\newblock Multiplicative renormalization and {H}opf algebras.
\newblock In O.~Ceyhan, Y.-I. Manin, and M.~Marcolli, editors, {\em Arithmetic
  and geometry around quantization}. Birkh\"auser Verlag, Basel, 2008.
\newblock [arXiv:0707.0555].

\bibitem{Sui07c}
W.~D. van Suijlekom.
\newblock Renormalization of gauge fields using {H}opf algebras.
\newblock In J.~T. B.~Fauser and E.~Zeidler, editors, {\em Quantum Field
  Theory}. Birkh\"auser Verlag, Basel, 2008.
\newblock [arXiv:0801.3170].

\bibitem{Sui08}
W.~D. van Suijlekom.
\newblock The structure of renormalization {H}opf algebras for gauge theories
  {I}: {R}epresenting {F}eynman graphs on {BV}-algebras. To appear in {\it Commun. Math. Phys.}
\newblock arXiv:0807.0999.

\end{thebibliography}
\newcommand{\noopsort}[1]{}

\end{document}